\documentclass[conference]{IEEEtran}
\IEEEoverridecommandlockouts
\usepackage[utf8]{inputenc} 
\usepackage[T1]{fontenc}
\usepackage{url}              
\usepackage{cite}             

\usepackage[cmex10]{amsmath}  
\usepackage{amssymb,amsfonts}
\usepackage{mleftright}       
\mleftright                   
\usepackage{multicol}
\usepackage{graphicx}         
\usepackage{booktabs}         
\usepackage[scr=rsfs]{mathalpha}                              

\usepackage{multirow}
\usepackage{xcolor}
\usepackage{url}
\usepackage{algorithm,algorithmic}
\usepackage{graphicx}
\usepackage{textcomp}
\usepackage{xcolor}
\usepackage{siunitx}
\usepackage{booktabs}
\usepackage{url}
\usepackage{subcaption}
\usepackage{stfloats}
\usepackage[capitalize]{cleveref}
\usepackage{multirow}
\usepackage{upgreek}
\usepackage{mathtools}
\mathtoolsset{centercolon}
\usepackage{pgfplots}\pgfplotsset{compat=1.9}
\pgfplotsset{compat=newest}
\newlength\fwidth

\usepackage{csquotes}
\usepackage{thmtools}
\usepackage{thm-restate}

\pgfdeclarelayer{background}
\pgfsetlayers{background,main}
\usepgfplotslibrary{groupplots}

\definecolor{plotclblue}{RGB}{57,106,180}
\colorlet{plotclblue}{black!20!plotclblue}
\definecolor{plotclorange}{RGB}{218,134,48}
\colorlet{plotclorange}{black!5!plotclorange}
\definecolor{plotclgreen}{RGB}{62,150,81}
\colorlet{plotclgreen}{black!10!plotclgreen}
\definecolor{plotclred}{RGB}{204,37,41}
\colorlet{plotclred}{black!10!plotclred}
\definecolor{plotclblack}{RGB}{83,81,84}
\colorlet{plotclblack}{black!10!plotclblack}
\definecolor{plotclviolet}{RGB}{107,76,154}
\colorlet{plotclviolet}{black!10!plotclviolet}
\definecolor{plotclredbrown}{RGB}{146,36,40}
\colorlet{plotclredbrown}{black!10!plotclredbrown}
\definecolor{plotclocer}{RGB}{148,139,61}
\colorlet{plotclocer}{black!10!plotclocer}
\definecolor{plotclyellow}{RGB}{255,255,30}
\colorlet{plotclyellow}{black!20!plotclyellow}
\definecolor{plotclcyan}{rgb}{000000,0.70000,0.70000}
\colorlet{plotclcyan}{black!20!plotclcyan}

\usepackage{tikz}
\usetikzlibrary{calc}
\usetikzlibrary{plotmarks}
\usetikzlibrary{decorations.pathmorphing}
\usetikzlibrary{decorations.markings}
\usetikzlibrary{decorations.pathreplacing}
\usetikzlibrary{positioning}
\usetikzlibrary{fit}
\usetikzlibrary{arrows}
\usetikzlibrary{patterns}
\usepackage{circuitikz}
\usetikzlibrary{shapes.geometric}
\newcommand{\norm}[1]{\left\lVert#1\right\rVert}
\DeclareMathSymbol{\shortminus}{\mathbin}{AMSa}{"39}
\makeatletter
\tikzset{
	nomorepostactions/.code={\let\tikz@postactions=\pgfutil@empty},
	mymark/.style 2 args={decoration={markings,
			mark= between positions 0 and 1 step (1/8)*\pgfdecoratedpathlength with{%
				\tikzset{#2,scale=1,every mark}\tikz@options
				\pgftransformresetnontranslations
				\pgfuseplotmark{#1}%
			},  
		},
		postaction={decorate},
		/pgfplots/legend image post style={
			mark=#1,mark options={#2, scale=1},every path/.append style={nomorepostactions}
		},
	},
}

\tikzset{%
	block-common/.style={draw, semithick, fill=white, minimum height=2em, minimum width=2em},
	block/.style={rectangle, block-common},
	txtblock/.style={block, align=center, minimum height=4em},
	txtblocktight/.style={block, align=center, minimum height=2.3em},
	input/.style={inner sep=1pt},
	output/.style={inner sep=1pt},
	sum/.style = {draw, fill=white, circle, minimum size=1.1em, inner sep=0pt,
		font={\small$+$}},
	prod/.style = {draw, fill=white, circle, minimum size=1.1em, inner sep=0pt,
		font={\normalsize$\times$}},
	pinstyle/.style = {pin edge={to-,thin,black}}
}

\usepackage{pgf, tikz, color, xcolor, soul}
\usetikzlibrary{shapes.geometric,arrows,positioning}
\usepackage{ifthen}
\usepackage{amsfonts, amsmath, amssymb, amsthm}
\newtheorem{theorem}{Theorem}
\newtheorem{definition}{Definition}

\newtheorem{lemma}[theorem]{Lemma}

\newtheorem{remark}{Remark}

\def\C{{\mathcal C}}
\def\D{{\mathcal D}}

\def\U{{\mathcal U}}

\def\X{{\mathcal X}}

\newcommand{\y}{\mathbf{y}}
\newcommand{\x}{\mathbf{x}}

\renewcommand{\u}{\mathbf{u}}
\renewcommand{\v}{\mathbf{v}}

\newcommand{\accept}{\mathrm{accept}}
\newcommand{\reject}{\mathrm{reject}}

\usepackage[draft,nomargin,marginclue,inline]{fixme}

\renewcommand*\FXLayoutMarginClue[3]{%
	\marginpar[%
	\raggedleft\@fxuseface{margin}\textcolor{black}{\ignorespaces#3 fixme}]{%
		\raggedright\@fxuseface{margin}\textcolor{black}{\ignorespaces#3 fixme}}}
\makeatother

\def \sG {\mathcal{G}}

\newtheorem{proposition}{Proposition}

\begin{document}
\title{Deterministic Identification Codes for Fading Channels\thanks{
The authors received support from the German Federal Ministry of Education and Research (BMBF) through the following projects: QTOK (Grants 16KISQ037K and 16KISQ038), QD-CamNetz (Grants 16KISQ077 and 16KISQ169), QUIET (Grants 16KISQ093 and 16KISQ0170), QTREX (Grants 16KISR026 and 16KISR038), QR.N (Grants 16KIS2195 and 16KIS2196), and QuaPhySI (Grants 16KIS1598K and 16KIS2234). Additionally, the authors acknowledge financial support from the BMBF under the program Souverän. Digital. Vernetzt. as part of the research HUB 6G-life (Project Identification Numbers: 16KISK002 and 16KISK263). H. Boche was further supported in part by the BMBF within the national initiative on Post Shannon Communication (NewCom) under Grant 16KIS1003K. C. Deppe was also supported by the DFG within the project DE1915/2-1. I.}%
}

\author{
\IEEEauthorblockN{Ilya Vorobyev\IEEEauthorrefmark{1}, Christian Deppe\IEEEauthorrefmark{3}\IEEEauthorrefmark{7}, and Holger Boche\IEEEauthorrefmark{2}\IEEEauthorrefmark{4}\IEEEauthorrefmark{5}\IEEEauthorrefmark{6}\IEEEauthorrefmark{7}}

\IEEEauthorblockA{
\IEEEauthorrefmark{1}\emph{IOTA Foundation}, Berlin\\
\IEEEauthorrefmark{2}\emph{Chair of Theoretical Information Technology}, Technical University of Munich\\
\IEEEauthorrefmark{3}Technical University of Braunschweig\\
\IEEEauthorrefmark{4}Cyber Security in the Age of Large-Scale Adversaries –
Exzellenzcluster, Ruhr-Universit\"at Bochum, Germany\\
\IEEEauthorrefmark{5}BMBF Research Hub 6G-life, Germany, 
\IEEEauthorrefmark{6} {\color{black}{Munich Quantum Valley (MQV)}} \\
\IEEEauthorrefmark{7}{\color{black}{ Munich Center for Quantum Science and Technology (MCQST) }}\\
Email: vorobyev.i.v@yandex.ru, christian.deppe@tu-braunschweig.de, boche@tum.de}
}
\maketitle

\begin{abstract}
   Many communication applications incorporate event-triggered behavior, where the conventional Shannon capacity may not effectively gauge performance. Consequently, we advocate for the concept of identification capacity as a more suitable metric for assessing these systems. We consider deterministic identification codes for the Gaussian AWGN, the slow fading, and the fast fading channels with power constraints. We prove lower bounds on capacities for the slow and the fast fading channels with side information for a wide range of fading distributions. Additionally, we present the code construction with efficient encoding which achieves the lower bound on capacity both for the slow and the fast fading channels.
   At last, we prove the same lower bound on the capacity of the fast fading channel without side information, i.e., the same lower bound holds even when the receiver does not know the fading coefficients.  
   As a result we show that compared with Shannon's message transmission paradigm we achieved completely different capacity scaling for deterministic identification codes for all relevant fading channels.
\end{abstract}

\section{Introduction}\label{sec::intro}


Numerous applications within the realm of Post Shannon communications \cite{schwenteck20236g,cabrera20216g}, envisioned for the forthcoming era of next-generation wireless networks (XG), are either founded upon or lead to event-triggered communication frameworks. The advancement to 6G technology promises an expansion in both IoT (Internet of Things) and Tactile Internet capabilities for consumers.  This suggests a scenario where both physical and virtual objects can be controlled over the network, facilitating the realization of autonomous systems~\cite{fettweis20216g,fettweis20226g}. Within the domain of communication systems, one pivotal measure of performance is ultra-reliability, ensuring data transmission with an extraordinarily low probability of failure~\cite{hassan2021key,adhikari20226g}. Notably, the development of methodologies to maintain outage probabilities below $10^{-5}$ holds significant importance in this context~\cite{bennis2018ultrareliable,salh2021survey}.
In these systems, the conventional measure of Shannon's message transmission capacity \cite{shannon1948mathematical} may not be suitable. Instead, the focus shifts to the concept of identification capacity as a pivotal quantitative metric. Particularly in scenarios involving event recognition, alarm triggering, or object detection, where the receiver's objective is to discern the presence or absence of a specific event, trigger an alarm, or confirm the existence of an object through a reliable binary decision, the identification capacity emerges as the primary performance indicator. The concept of identification also has the potential to meet the strict latency requirements~\cite{fettweis2014tactile} and significantly reduce energy consumption~\cite{schwenteck20236g}.

In contrast to the traditional Shannon communication paradigm \cite{shannon1948mathematical}, where the sender encodes messages for faithful reproduction by the receiver, the design of coding schemes in identification settings serves a different purpose. Here, the objective is to determine whether a specific message was transmitted or not. The concept of identification in communication theory, pioneered by Ahlswede and Dueck \cite{Idchannels}, incorporates randomization, wherein a randomized encoder selects codewords from a codebook composed of distributions over codewords. Notably, it has been observed that introducing local randomness at the encoder results in remarkably efficient identification, with the codebook size exhibiting a double-exponential growth relative to the codeword length. This stands in stark contrast to conventional message transmission problems \cite{shannon1948mathematical}, where the codebook size typically grows exponentially with the codeword length.

The motivation behind Ahlswede and Dueck's \cite{Idchannels} development of the Randomized Identification (RI) problem can be traced back to JaJa's work \cite{Jaja} on Deterministic Identification (DI), which considers codewords determined by deterministic functions of the messages from a communication complexity perspective. Ahlswede and Dueck aimed to demonstrate that leveraging randomness, akin to techniques in communication complexity, provides a significant advantage over the DI problem, yielding an exponential gap in codebook size. However, in complexity-constrained applications, DI may be preferred over RI.

Further exploration of RI code construction is detailed in \cite{VerduWei93,Derebeyo_lu_2020, von2023beyond}, while studies on RI for Gaussian channels are presented in \cite{burnashev2000identification, labidi2021identification, ezzine2020common}. Deterministic codes offer advantages such as simpler implementation, explicit construction, and reliable single-block performance. Ahlswede and Cai investigate DI for compound channels in \cite{ahlswede1999identification}. Remarkably, in \cite{labidi2021identification} it was established that the DI capacity is infinite irrespective of scaling considerations for the codebook size for the case of non-discrete additive white noise and noiseless feedback under both average and peak power constraints. 

The most basic channel model commonly used in practice is the Additive White Gaussian Noise (AWGN) model, where the received signal is presumed to undergo solely constant attenuation and delay \cite{fano1968transmission}. However, for digital transmission, a more sophisticated model is often required. In these scenarios, it becomes essential to consider additional propagation effects known as fading, which alters the envelope of the received signal. Depending on the statistical characteristics of fading, it is classified as either a fast or slow fading channel. The exploration of DI codes for such channels was started in~\cite{salariseddigh2021deterministic}. It was established that for fading channels the size of the code grows as $n^{C_{ID}n}$, and some bounds for the capacity $C_{ID}$ were established. We continue this line of research and prove new lower and upper bounds on capacity.


\subsection{Our contribution}
The first contribution of the paper is the construction of codes with certain properties. These codes can be efficiently constructed and encoded, i.e., with a complexity polynomial in length $n$. These codes are used to prove lower bounds on the capacities of deterministic identification codes.

The next contribution of this paper concerns the capacity of the slow fading channel with side information. In this paper side information means the knowledge of fading coefficients at the receiver. We introduce the definition of $\eta$-outage capacity and prove the lower bound $1/4$. The lower bound is achieved with a code construction mentioned above, i.e., codes that achieve this lower bound can be efficiently constructed and encoded. Even though the bound looks the same as the bound from~\cite{salariseddigh2021deterministic}, it can be applied to a wider class of fading distributions. This improvement is achieved, in part, since we consider another definition of a DI code. The detailed discussion on various definitions of DI code for the slow fading channel is in Subsection~\ref{subsec::discussion on definitions}.

Another important result is the lower bound $1/4$ on the capacity of fast fading channels with channel side information for a wide class of fading distributions. This bound is also constructive. Surprisingly, the lower bound does not change even if the probability that the fading coefficient equals zero is positive as long as it is less than one. The result applies to the practical fading distributions as well.

Moreover, we prove a new lower bound for the fast fading channels without channel side information, when the receiver does not know the fading coefficients. It is proved that we can achieve the same lower bound as in the case with channel side information. We prove the lower bound under the assumption that the mathematical expectation of the fading coefficient is not zero. It is a major restriction, which cuts off some important distributions. 

Finally, we refine the existing upper bound on the capacity of the Gaussian, slow fading, and fast fading channels.

\subsection{Outline}
The rest of the paper is structured as follows. In Section~\ref{sec::notations and defs} we give notations and key definitions. In Section~\ref{sec::new results} we formulate and discuss our new results. Section~\ref{sec::construction} describes an explicit code construction that is later used to prove lower bounds on the identification capacity. Section~\ref{sec::upp_bound_gauss} contains the improved upper bound on the capacity of the Gaussian AWGN channel. In Section~\ref{sec::fading with csi} the bounds on the capacities of the slow and fast fading channels with side information are proved. In Section~\ref{sec::fast fading without csi} we prove a lower bound on the capacity of the fast fading channel without side information. Section~\ref{sec::conclusion} contains final remarks and discussion about future research.

\section{Notations and definitions}\label{sec::notations and defs}
The distribution of a random variable (RV) $X$ is denoted by $P_X$; 
$\X^c$ denotes the complement of $\X$; if $X$ is a RV with distribution $P_X$, we denote the expectation of $X$ by $\mathbb{E} X$ and by ${\text{Var}}X$ the variance of $X$; the entropy of a random variable $X$ is denoted as $H(X)$; binary entropy function is denoted as $h(x)$. Bold symbols $\x$ are used to denote vectors; $\norm{\x}_p$ denotes $p$-norm of a vector; $[M]$ denotes the set $\{1, 2, \ldots, M\}$; $cl(H)$ denotes the closure of the set $H$.
All logarithms and information quantities are taken to the base $2$.

We consider the Gaussian AWGN channel, the slow fading channel, and the fast fading channels. 
All these channels can be described by equations
\begin{equation}
    y_t=h_tx_t+z_t,
\end{equation}
where $\{z_t\}_{t=1}^n$ is a sequence of i.i.d. random variables with normal distribution $N(0, \sigma^2)$, $\{h_t\}_{t=1}^n$ is a sequence of random variables with an arbitrary distribution. In the case of the slow fading $h_i=h$ for all $i\in [1, n]$, for the case of the fast fading $h_i$ is a sequence of i.i.d. random variables. For the Gaussian AWGN channel $h_i=1$. We introduce power constraints, i.e., all codewords $\x$ should satisfy inequality $\norm{\textbf{x}}_2^2\leq nA$.

We distinguish channels with and without side information. In the case of channels with side information the receiver knows fading coefficients, so the decoding regions can depend on the realization of random variables $h_i$.

This paper considers codes whose size scales as $L_{DI}(n, R) = n^{nR}$. That is different from the standard scaling in coding theory $L_{Tr}(n, R)=2^{nR}$ and the scaling used for randomized identification codes $L_{RI}(n, R)=2^{2^{Rn}}$. The scaling $L_{DI}(n, R)=n^{nR}$ is chosen since we can obtain non-trivial bounds on the code rate for such scaling. In other words, the maximal size of the deterministic identification codes for fading channels with error probabilities tending to zero grows as $L_{DI}(n, R)=n^{nR}$. This also means that the code size is significantly bigger than the sizes of the codes used in usual Shannon coding problems but not as big as in randomized identification. In other words, if we fix code size $M$ then for deterministic identification the length of the code $n$ would be significantly smaller than the length of the code of the same size for transmission. This also means that such a deterministic identification code will require a significantly smaller amount of energy since it is proportional to the length of the code.

\subsection{Gaussian DI Code}
\begin{definition}
\label{GdeterministicIDCode}
A $(L_{DI}(n, R)=n^{nR},n)$ DI code for a Gaussian channel $\sG$ under input constraint $A$, assuming $n^{nR}$ is an integer, is defined as a system $(\U,\mathscr{D})$ consisting of a codebook $\U=\{ \mathbf{u}_i \}_{i\in[n^{nR}]}$, $\U\subset \mathbb{R}^n$, such that
\begin{align} 
    \norm{\mathbf{u}_i}_2^2 \leq nA \;,\,
\end{align}
for all $i\in[n^{nR}]$ and a collection of decoding regions
$\mathscr{D}=\{ \D_i \}_{i\in[n^{nR}]}$
with
\begin{align}
    \bigcup_{i=1}^{n^{nR}} \D_i \subset \mathbb{R}^n \;.\,
\end{align}

The error probabilities of the identification code $(\U,\mathscr{D})$ are given by
\begin{align}
 P_{e,1}(i)&= 1-\int_{\D_i} f_{\mathbf{Z}}(\mathbf{y}-\mathbf{u}_i)\, d\mathbf{y}
  \;,\,
 \label{Eq.GTypeIErrorDef}
 \\
 P_{e,2}(i,j)&= \int_{\D_j} f_{\mathbf{Z}}(\mathbf{y}-\mathbf{u}_i) \, d\mathbf{y} &&\hspace{-2cm} 
 \;.\,
 \label{Eq.GTypeIIErrorDef}
\end{align}
with the probability density function of the noise
\begin{align}
   f_{\mathbf{Z}}(\mathbf{z})=\frac{1}{(2\pi\sigma^2)^{n/2}}
e^{-\norm{\mathbf{z}}_2^2/2\sigma^2} \;\, .
\end{align}
A $(n^{nR},n,\lambda_1,\lambda_2)$ Gaussian DI code satisfies
\begin{align}
    \label{Eq.GTypeIError}
    P_{e,1}(i) &\leq \lambda_1 \;,\,
    \\
    \label{Eq.GTypeIIError}
    P_{e,2}(i,j) &\leq \lambda_2 \;,\,
\end{align}
for all $i,j\in [n^{nR}]$, such that
$i\neq j$.

A rate $R>0$ is called achievable if for every
$\lambda_1,\lambda_2>0$ and sufficiently large $n$, there exists a $(n^{nR},n,\lambda_1,\lambda_2)$ DI code. 
The operational DI capacity of the Gaussian channel is defined as the supremum of achievable rates, and will be denoted by $C_{DI}(\mathcal{G})$. 
\end{definition}
We emphasize that we don't require decoding regions $\D_i$ to be disjoint, not in this definition or any other definition from the paper.

The DI code is used as depicted in \cref{fig:comm_sys}. Given an identity $i\in [n^{nR}]$ the encoder transmits codeword $\u_i$. The receiver with identity $j$ accepts the identity if and only if the output $\y\in D_j$. So, the missed-identification error or type I error corresponds to the case, when $i=j$, but the identity is rejected by the receiver. The probability of such error is given by \eqref{Eq.GTypeIErrorDef}. The false identification error or type II error corresponds to the case when the identity $i$ is accepted by the receiver expecting the identity $j$. The probability of such error is defined in \eqref{Eq.GTypeIIErrorDef}.

\begin{figure}[t]
\setlength\fwidth{0.5\linewidth}
    \centering
    {  \begin{tikzpicture}[font=\footnotesize, >=latex]
    \node[input] (M) {};
    \node[block, align=center, right=1.8em of M]
        (id-layer-1)
        {\scriptsize DI\\ \scriptsize encoder};
    \node[block, right=1.8em of id-layer-1, align=center]
        (ch-bob)
        {\scriptsize Channel};
    \node[block, align=center, right=1.8em of ch-bob]
        (id-layer-2)
        {\scriptsize DI\\ \scriptsize verifier};   
    \node[output, right=1.8em of id-layer-2] (M-hat) {};

    \draw[->] (M)
        -- node[xshift=-2.2em, align=center]   {Transmitter's\\identity, $i$} (id-layer-1);
    \draw[->] (id-layer-1)
        -- coordinate (mid) node[above] {$\u_i$} (ch-bob);
    \draw[->] (ch-bob)
        -- node[above, xshift=-0.3em]   {$y$} (id-layer-2);
    \draw[->] (id-layer-2)
        -- node[xshift=2.2em, align=center]    {\{$\accept$, \\$\reject$\}} (M-hat);    
    \draw[<-] (id-layer-2)
        -- ++(0,-0.8) node [below, align=center] {\scriptsize Receiver's identity, $j$};
    
    \node[draw, rectangle, densely dashdotted, inner ysep=0.5em, yshift=-0.1em, inner xsep=0.4em, xshift=-0.0em,
    fit=(id-layer-1)] (alice) {};
    \node[draw, rectangle, densely dashdotted, inner xsep=0.4em,
    inner ysep=0.5em, yshift=-0.1em, xshift=-0.0em,
    fit=(id-layer-2)] (bob) {};
      
    \node[anchor=south] at (alice.north) {Transmitter};
    \node[anchor=south] at   (bob.north) {Receiver};

  \end{tikzpicture}}
    \caption{Channel model for deterministic identification.}
    \label{fig:comm_sys}
    \vspace*{-1.em}
\end{figure}
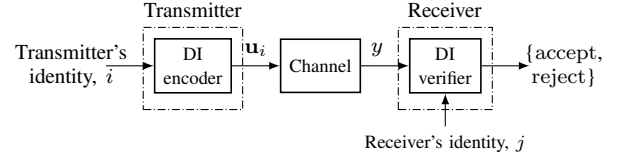

\subsection{Fading Channels with Channel Side Information}
Below we give the definitions of slow and fast fading channels with channel side information.
In this paper side information means the knowledge of fading coefficients at the receiver.
\begin{definition}
        A $(L_{DI}(n, R)=n^{nR}, n)$ DI code for a slow  fading channel $\mathcal{G}_{slow}$ with channel side information (CSI) under input constraints $A$, assuming $n^{nR}$ is an integer, is defined as a system $(\mathcal{U}, \mathcal{D})$ that consists of a codebook $\mathcal{U}=\{\u_i\}_{i\in[n^{nR}]}$, $\mathcal{U}\subset \mathbb{R}^n$, such that
\begin{equation}
    \|\u_i\|_2^2\leq An, \text{ for all } i\in [n^{nR}],
\end{equation}
and a collection of decoding regions $\mathcal{D}=\{D_{i, h}\}_{i\in [n^{nR}], h\in\mathbb{R}}$, $D_{i, h}\subset \mathbb{R}^n$. The error probabilities are given by
\begin{align}
    P_{e, 1}(i, h)&=\int\limits_{D^c_{i, h}}f_{\mathbf{z}}(\mathbf{y} - h \cdot \u_i) d\mathbf{y},\\
    P_{e, 2}(i, j, h)&=\int\limits_{D_{j, h}}f_{\mathbf{z}}(\mathbf{y} - h\cdot \u_i) d\mathbf{y}.
\end{align}

A $(n^{nR}, n, \lambda_1, \lambda_2, \eta)$ DI code satisfies
\begin{align}
    \sup\limits_{h\in H_0}P_{e, 1}(i, h)&\leq \lambda_1,\\
    \sup\limits_{h\in H_0}P_{e, 2}(i, j, h)&\leq \lambda_2,\\
    \mathbb{P}(h\in H_0)&\geq 1-\eta
\end{align}
for some set $H_0\subset \mathbb{R}$ and for all $i, j \in [n^{nR}]$, $i\ne j$. The number $\eta$ is called an outage probability.
\end{definition}

\begin{definition}
        A $(L_{DI}(n, R)=n^{nR}, n)$ DI code for a fast fading channel $\mathcal{G}_{fast}$ with channel side information (CSI) under input constraints $A$, assuming $n^{nR}$ is an integer, is defined as a system $(\mathcal{U}, \mathcal{D})$ that consists of a codebook $\mathcal{U}=\{\u_i\}_{i\in[n^{nR}]}$, $\mathcal{U}\subset \mathbb{R}^n$, such that
\begin{equation}
    \|\u_i\|^2\leq An, \text{ for all } i\in [n^{nR}],
\end{equation}
and a collection of decoding regions $\mathcal{D}=\{D_{i, \mathbf{h}}\}_{i\in [n^{nR}], \mathbf{h}\in\mathbb{R}^n}$, $D_{i, \mathbf{h}}\subset \mathbb{R}^n$. The error probabilities are given by
\begin{align}
    P_{e, 1}(i)&=\int\limits_{\mathbb{R}^n}\int\limits_{D^c_{i, \mathbf{h}}}f_{\mathbf{z}}(\mathbf{y} - \text{diag}(\mathbf{h}) \u_i) d\mathbf{y}dP_{\mathbf{h}},\\
    P_{e, 2}(i, j)&=\int\limits_{\mathbb{R}^n}\int\limits_{D_{j, \mathbf{h}}}f_{\mathbf{z}}(\mathbf{y} - \text{diag}(\mathbf{h}) \u_i) d\mathbf{y}dP_{\mathbf{h}},
\end{align}
where $P_{\mathbf{h}}$ is a probability distribution of a vector of RV $\mathbf{h}=(h_1, \ldots, h_n)$.
 
A $(n^{nR}, n, \lambda_1, \lambda_2)$ DI code satisfies
\begin{align}
    P_{e, 1}(i)&\leq \lambda_1,\\
    P_{e, 2}(i, j)&\leq \lambda_2,
\end{align}
for all $i, j \in [n^{nR}]$, $i\ne j$.
\end{definition}

A rate $R>0$ of DI code for the slow  fading channel is called achievable for the outage probability $\eta$ if for any $\lambda_1, \lambda_2>0$ and sufficiently large $n$ there exists a $(n^{Rn}, n, \lambda_1, \lambda_2, \eta)$ DI code. 
A rate $R>0$ of DI code for the fast  fading channel is called achievable if for any $\lambda_1, \lambda_2>0$ and sufficiently large $n$ there exists a $(n^{Rn}, n, \lambda_1, \lambda_2)$ DI code. 

Capacities are defined as supremums of corresponding rates and denoted as $C^{CSI}_{DI}(\mathcal{G}_{slow}, \eta)$ and $C_{DI}^{CSI}(\mathcal{G}_{fast})$.

\subsection{Discussion about alternative definitions of the capacity of the slow fading channel}\label{subsec::discussion on definitions}

In the paper\cite{salariseddigh2021deterministic} the authors assume that $h$ is a continuous random variable with probability density function $f_h(x)$ and finite moments, and $h$ belongs to bounded set $H$ with probability $1$. They consider the following definition of DI codes for the slow fading channel.

\begin{definition}[\cite{salariseddigh2021deterministic}]
        A $(L_{DI}(n, R)=n^{nR}, n)$ DI-sup code for a slow fading channel $\mathcal{G}_{slow}$ with channel side information (CSI) under input constraints $A$, assuming $n^{nR}$ is an integer, is defined as a system $(\mathcal{U}, \mathcal{D})$ that consists of a codebook $\mathcal{U}=\{\u_i\}_{i\in[n^{nR}]}$, $\mathcal{U}\subset \mathbb{R}^n$, such that
\begin{equation}
    \|\u_i\|_2^2\leq An, \text{ for all } i\in [n^{nR}],
\end{equation}
and a collection of decoding regions $\mathcal{D}=\{D_{i, h}\}_{i\in [n^{nR}], h\in\mathcal{H}}$, $D_{i, h}\subset \mathbb{R}^n$. The error probabilities are given by
\begin{align}
    P_{e, 1}(i, h)&=\int\limits_{D^c_{i, h}}f_{\mathbf{z}}(\mathbf{y} - h \cdot \u_i) d\mathbf{y},\\
    P_{e, 2}(i, j, h)&=\int\limits_{D_{j, h}}f_{\mathbf{z}}(\mathbf{y} - h\cdot \u_i) d\mathbf{y}.
\end{align}

A $(n^{nR}, n, \lambda_1, \lambda_2)$ DI-sup code satisfies
\begin{align}
    \sup\limits_{h\in H}P_{e, 1}(i, h)&\leq \lambda_1,\\
    \sup\limits_{h\in H}P_{e, 2}(i, j, h)&\leq \lambda_2,
\end{align}
for all $i, j \in [n^{nR}]$, $i\ne j$.
\end{definition}

Another way to define DI code for the slow fading channel is to consider average errors over $h$ like in the case of the fast fading.

\begin{definition}
        A $(L_{DI}(n, R)=n^{nR}, n)$ DI-av code for a slow fading channel $\mathcal{G}_{slow}$ with channel side information (CSI) under input constraints $A$, assuming $n^{nR}$ is an integer, is defined as a system $(\mathcal{U}, \mathcal{D})$ that consists of a codebook $\mathcal{U}=\{\u_i\}_{i\in[n^{nR}]}$, $\mathcal{U}\subset \mathbb{R}^n$, such that
\begin{equation}
    \|\u_i\|_2^2\leq An, \text{ for all } i\in [n^{nR}],
\end{equation}
and a collection of decoding regions $\mathcal{D}=\{D_{i, h}\}_{i\in [n^{nR}], h\in\mathbb{R}}$, $D_{i, h}\subset \mathbb{R}^n$. The error probabilities are given by
\begin{align}
    P_{e, 1}(i)&=\int\limits_{\mathbb{R}}\int\limits_{D^c_{i, h}}f_{\mathbf{z}}(\mathbf{y} - h \cdot \u_i) d\mathbf{y}dP_h,\\
    P_{e, 2}(i, j)&=\int\limits_{\mathbb{R}}\int\limits_{D_{j, h}}f_{\mathbf{z}}(\mathbf{y} - h\cdot \u_i) d\mathbf{y}dP_h.
\end{align}

A $(n^{nR}, n, \lambda_1, \lambda_2)$ DI-av code satisfies
\begin{align}
    P_{e, 1}(i)&\leq \lambda_1,\\
    P_{e, 2}(i, j)&\leq \lambda_2,
\end{align}
for all $i, j \in [n^{nR}]$, $i\ne j$.
\end{definition} 

Define achievable rates and capacities for these two definitions in a standard way. Denote the capacities as $C^{CSI}_{DI-sup}(\mathcal{G}_{slow})$ and $C_{DI-av}^{CSI}(\mathcal{G}_{slow})$.

This definition of DI-sup codes considers the supremum over the set of all values $H$, whereas we look at the supremum over some set $H_0$ such that $\mathbb{P}(h\in H_0)\geq 1-\eta$. It is very similar to our definition for $\eta=0$, although even in this case we can exclude from $H$ some subset of probability $0$. It is obvious that a $(n^{nR}, n, \lambda_1, \lambda_2)$ DI-sup code is also a $(n^{nR}, n, \lambda_1, \lambda_2, \eta)$ DI code for any $\eta$ and $(n^{nR}, n, \lambda_1, \lambda_2)$ DI-av code. It is also easy to see that any $(n^{nR}, n, \lambda_1, \lambda_2, \eta)$ DI code is a $(n^{nR}, n, \lambda_1+\eta, \lambda_2 +\eta)$ DI-av code. So, if we prove some lower bound on capacity $C^{CSI}_{DI}(\mathcal{G}_{slow}, \eta)$ for arbitrary $\eta$, the same lower bound is true for $C_{DI-av}^{CSI}(\mathcal{G}_{slow})$.

The disadvantage of the definition of DI-sup code is that it depends only on the support set $h$ of random variable $h$, but not on the distribution of $h$. It seems that the requirements in this definition are too strict. The definition we propose makes use of the additional parameter $\eta$ and therefore is more flexible.

The definition of DI-av codes seems to be not very adequate from the practical point of view. For the slow fading the state of the channel, i.e., realization of the random variable $h$, can be the same for long periods. Therefore, the average error is not a good predictor of the actual behavior of the channel.

\subsection{Fast Fading Channel without Side Information}

\begin{definition}
        A $(L_{DI}(n, R)=n^{nR}, n)$ DI code for a fast fading channel $\mathcal{G}_{fast}$ under input constraints $A$, assuming $n^{nR}$ is an integer, is defined as a system $(\mathcal{U}, \mathcal{D})$ that consists of a codebook $\mathcal{U}=\{u_i\}_{i\in[n^{nR}]}$, $\mathcal{U}\subset \mathbb{R}^n$, such that
\begin{equation}
    \norm{\u_i}^2\leq An, \text{ for all } i\in [n^{nR}],
\end{equation}
and a collection of decoding regions $\mathcal{D}=\{D_{i}\}_{i\in [n^{nR}]}$, $D_{i}\subset \mathbb{R}^n$. The error probabilities are given by
\begin{align}
    P_{e, 1}(i)&=\int\limits_{\mathbb{R}^n}\int\limits_{D^c_{i }}f_{\mathbf{z}}(\mathbf{y} - \text{diag}(\mathbf{h}) \u_i) d\mathbf{y}dP_{\mathbf{h}},\\
    P_{e, 2}(i, j)&=\int\limits_{\mathbb{R}^n}\int\limits_{D_{j}}f_{\mathbf{z}}(\mathbf{y} - \text{diag}(\mathbf{h}) \u_i) d\mathbf{y}dP_{\mathbf{h}},
\end{align}
where $P_{\mathbf{h}}$ is a probability distribution of a vector of RV $\mathbf{h}=(h_1, \ldots, h_n)$.
 
A $(n^{nR}, n, \lambda_1, \lambda_2)$ DI code satisfies
\begin{align}
    P_{e, 1}(i)&\leq \lambda_1,\\
    P_{e, 2}(i, j)&\leq \lambda_2,
\end{align}
for all $i, j \in [n^{nR}]$, $i\ne j$.
\end{definition}

Note that in these definitions decoding regions do not depend on $\mathbf{h}$.
 
A rate $R>0$ of DI code for the slow or fast fading channel without side information is called achievable if for any $\lambda_1, \lambda_2>0$ and sufficiently large $n$ the exists a corresponding $(n^{Rn}, n, \lambda_1, \lambda_2, \eta)$ or $(n^{Rn}, n, \lambda_1, \lambda_2)$ DI code. 

Capacity is defined as the supremum of achievable rates and denoted as $C_{DI}(\mathcal{G}_{fast})$.

\section{New Results and Discussion}\label{sec::new results}

\subsection{Code Construction}
 We provide a code construction with the following properties.
\begin{theorem}\label{th::construction}
     For any positive constants $a$, $A$, $B$, $a<1/8$, and any $n>n_0(a)$ there exists a code $\mathcal{U}(n)=\{\u_i\}_{i\in[M]}$, $\u_i\in \mathbb{R}^n$ of size $M\geq 2^{\left(\frac{1}{4}-2a\right)n\log_2n}$
    with the following properties.
    \begin{enumerate}
        \item  $\norm{\u_i}_2^2\leq An$ for all $i$.
        \vspace{1mm}
        \item  $\norm{\u_i}_4^4\leq Bn$ for all $i$.
        \vspace{1mm}
        \item $\norm{\u_i-\u_j}_2\geq n^{1/4+a}$ for all $i\neq j$.
    \end{enumerate}

    Moreover, such code can be efficiently constructed and encoded, i.e., the time complexity of construction and encoding are polynomial.
\end{theorem}

A code with such properties can be used as DI code in Theorem~\ref{th::slow fading} and in the third claim of Theorem~\ref{th::fast fading}. Obviously, the lower bound for the Gaussian AWGN channel is also achieved with such codes.
The code construction is based on the concatenation of two Reed-Solomon codes and can be easily implemented in practice.

\subsection{Upper Bound for the Gaussian AWGN channel}
For the capacity of Gaussian AWGN channel $\mathcal{G}$ with power constraints lower and upper bounds have been proved in~\cite{salariseddigh2021deterministic}
\begin{equation}
    \frac{1}{4}\leq C_{DI}(\mathcal{G})\leq 1.
\end{equation}
Note that the bounds do not depend on $A$. The bounds for the capacities of fading channels also would not depend on $A$. However, for finite length the parameter $A$ is important.
We refine the upper bound and prove the following.

\begin{theorem}\label{th::upper bound for gaussian channel}
    For a Gaussian AWGN channel $\mathcal{G}$ with power constraints 
    \begin{equation}
    C_{DI}(\mathcal{G})\leq \frac{1}{2}.
    \end{equation}
\end{theorem}

\begin{remark}
    This result was also independently obtained by J. Huffmann, though it remains unpublished (Personal communication).
\end{remark}

This result can be extended to the case of the slow fading and the fast fading channels.

\subsection{Slow Fading Channel with CSI}

In paper~\cite{salariseddigh2021deterministic} the authors proved the following theorem 

\begin{theorem}\label{th::slow fading Mohammad}
    If $0 \in cl(H)$ then 
    \begin{equation}
        C^{CSI}_{DI-sup}(\mathcal{G}_{slow})=0.
    \end{equation}
    If $0\notin cl(H)$ then
    \begin{equation}
        \frac{1}{4}\leq C^{CSI}_{DI-sup}(\mathcal{G}_{slow})\leq 1.
    \end{equation}    
\end{theorem}

We are going to prove the following.
\begin{theorem}\label{th::slow fading}
    1. For any $\eta$, $\eta<\mathbb{P}(h=0)$, the capacity of the slow fading channel is 0, i.e.,
    \begin{equation}
        C_{DI}^{CSI}(\mathcal{G}_{slow}, \eta)=0.
    \end{equation}
    2. For any $\eta$, $ \mathbb{P}(h=0)<\eta<1$, the capacity of the slow fading channel satisfies the inequality
    \begin{equation}
        \frac{1}{4}\leq C_{DI}^{CSI}(\mathcal{G}_{slow}, \eta)\leq \frac{1}{2}.
    \end{equation}
    Moreover, codes that achieve this lower bound can be efficiently constructed and encoded, i.e., the time complexity of construction and encoding are polynomial.
\end{theorem}

The improvement of the upper bound from $1$ to $1/2$ is based on  Theorem~\ref{th::upper bound for gaussian channel}, and can be done for any of the alternative definitions of the capacity mentioned in Section~\ref{sec::notations and defs}.

As for the lower bound, we note that our result is constructive, while Theorem~\ref{th::slow fading Mohammad} was proved with probabilistic method. Additionally, in the case when $0\in cl(H)$ we can give some meaningful bounds for some $\eta$, whereas for the definition from~\cite{salariseddigh2021deterministic} it can only be said that the capacity is zero. It is an interesting and practically relevant improvement, since for many fading distributions, used in practice, $0$ belongs to the closure of $H$. In particular, this is true for Rayleigh, Rician, and Nakagami distributions, which are used to describe the distribution of $|h|$. However, the motivation for the usage of these distributions comes from the model with $h\in \mathbb{C}$, for which our results can not be applied directly.

For the capacity with average error $C^{CSI}_{DI-av}(\mathcal{G}_{slow})$ the lower bound $1/4$ can be proved for the case $\mathbb{P}(h=0)=0$. In case if $\mathbb{P}(h=0)>0$ the capacity $C^{CSI}_{DI-av}(\mathcal{G}_{slow})=0$.

Recall~\cite{tse2005fundamentals} that the transmission $\varepsilon$-outage capacity $C_{\varepsilon}$ for the slow fading channel equals
\begin{equation}
    C_{\varepsilon} = \log(1+F^{-1}(1-\varepsilon)\text{SNR}),
\end{equation}
where $F(x)=\mathbb{P}(|h|^2>x)$. A similar result was obtained in~\cite{ezzine2023message} for MIMO slow fading channels.
We emphasize that $C_{\varepsilon}$ depends both on power constraints and on the distribution of fading coefficients, whereas the identification capacity does not depend on power constraints at all, and the dependence in the distribution of $h$ is limited to the fact if $\mathbb{P}(h=0)$ is smaller or greater than $\eta$. We also emphasize that the scaling for the identification is different.
As clearly mentioned in the introduction, ultra-reliability, i.e., maintaining very low probabilities of outage, is one of the central requirements for 6G networks. This fact signifies the importance of our result that deterministic identification allows us to achieve better performance with respect to outage capacity.

\subsection{Fast Fading Channel with CSI}

In paper~\cite{salariseddigh2021deterministic} the following result is proved.

\begin{theorem}\label{th::fast fading Mohammad}
    Assume that the fading coefficients are positive and bounded away from zero, i.e., $0\notin cl(H)$, $\mathbb{P}(h \in H) = 1$ for some $H$. Moreover, assume that $h$ is an absolutely continuous random variable with finite moments. Then
    \begin{equation}
        1/4\leq C_{DI}^{CSI}(\mathcal{G}_{fast})\leq 1.
    \end{equation}
\end{theorem}

We are going to prove a more general theorem.

\begin{theorem}\label{th::fast fading}
    
    1. If $\mathbb{E}h^2 < \infty$ then 
    \begin{equation}
        C_{DI}^{CSI}(\mathcal{G}_{fast})\leq 1/2.
    \end{equation}
    
    2. If $\mathbb{P}(h=0) = 0$ then 
    \begin{equation}
        C_{DI}^{CSI}(\mathcal{G}_{fast})\geq 1/4.
    \end{equation}

    3. If $\mathbb{P}(h=0) < 1$ and $\mathbb{E}h^4<\infty$ then 
    \begin{equation}
        C_{DI}^{CSI}(\mathcal{G}_{fast})\geq 1/4.
    \end{equation}
    Moreover, codes that achieve the lower bound in claim 3 can be efficiently constructed and encoded, i.e., the time complexity of construction and encoding are polynomial.
    
\end{theorem}

The improvement of the upper bound follows from Theorem~\ref{th::upper bound for gaussian channel}. Our lower bound in the third claim is constructive and is proved with a much smaller number of assumptions on the distribution $h$. The most important difference is that we do not need the distribution to be continuous, and do not need coefficients to be bounded away from zero. 

Recall~\cite{tse2005fundamentals} that the transmission capacity for the fast fading channel is equal to
\(
    C=\mathbb{E}(\log(1+|h|^2\text{SNR})).
\)
The transmission capacity is clearly influenced by both the SNR and the distribution of $h$. The identification capacity for the fast fading channel behaves differently and does not depend on power constraints, noise level, or distribution of $h$. Remarkably, even in the case of positive probability $\mathbb{P}(h=0)$ the same lower bound $1/4$ holds. Let us not forget that the scaling for identification is different. So we see the advantage of identification over transmission coding in the fast fading model as well.

\subsection{Fast Fading without CSI}

\begin{theorem}\label{th::fast fading without CSI}
    Let $\mathbb{E} h^4<\infty$, $\mathbb{E}h\neq 0$. Then 
    \begin{equation}
     C_{DI}(\mathcal{G}_{fast})\geq 1/4.   
    \end{equation}
\end{theorem}

This theorem tells us that the same bound $1/4$ is true even when the fading coefficients are not known for the receiver. This theorem requires $\mathbb{E}h\ne 0$, which is not true for many interesting distributions. That is why the calculation of the capacity $C_{DI}(\mathcal{G}_{fast})$ for the case $\mathbb{E}h=0$ is an important open question.

\section{Proof of Theorem~\ref{th::construction}}\label{sec::construction}

	 Consider a concatenation $\mathcal{C}$ of two Reed-Solomon codes. The inner code has alphabet size $q_1=n^{1/4-b}$, $b\in (a, 2a)$, $q_1$ is a prime power, length $n_1$, distance $d_1\geq \varepsilon_1\cdot n_1$, and size $M_1=q_1^{n_1(1-\varepsilon_1)}$.

    The parameters of the outer code are as follows: alphabet size $q_2=M_1$, length $n_2=n/n_1$, distance $d_2\geq \varepsilon_2\cdot n_2$, and size $M_2=q_2^{n_2(1-\varepsilon_2)}$.

    The concatenated code $\mathcal{C}$ has alphabet $q_1$, length $n$, distance $d\geq \varepsilon_1\varepsilon_2n$, and size $M=M_2=q_1^{n((1-\varepsilon_1)(1-\varepsilon_2)}$.

    Let $A'=\min(A, \sqrt{B})$. Consider a set of $q_1$ elements $\mathcal{Q}=\{-\sqrt{A'}+2\sqrt{A'}j/(q_1-1)\}$ for $j=0, 1, \ldots, q-1$.
    Construct a code $\mathcal{U}$ in $\mathbb{R}^n$ from the code $\mathcal{C}$ by using the set $\mathcal{Q}$ as an alphabet.

    The Euclidean distance between any two codewords in $\mathcal{U}$ is at least
    \begin{equation}
        \sqrt{d\cdot \frac{4A'}{(q_1-1)^2}}\geq 
        \sqrt{n^{1-1/2+2b}}\cdot 2\sqrt{A'\varepsilon_1\varepsilon_2}=n^{1/4+b+o(1)}.
    \end{equation}
    The input constraints $\norm{\u_i}_2^2\leq An$ and $\norm{\u_i}_4^4\leq Bn$ are obviously satisfied.
    The size of the code is 
    \begin{equation}
        M=q_1^{n(1-\varepsilon_1)(1-\varepsilon_2)}
        =
        n^{n(1/4-b)(1-\varepsilon_1)(1-\varepsilon_2)}.
    \end{equation}
    We can choose $\varepsilon_1, \varepsilon_2$ small enough so that $(1/4-b)(1-\varepsilon_1)(1-\varepsilon_2)>1/4-2a$.
    So, the code $\mathcal{U}$ satisfies the required conditions.

	\section{Proof of Theorem~\ref{th::upper bound for gaussian channel}}  \label{sec::upp_bound_gauss}
	
	The proof is based on the following idea: if we have small errors of the first and the second types, then any two codewords should be separated by some distance. We prove that the distance between any two codewords should grow to infinity as a function of $n$. After that, we prove an upper bound by using the sphere packing argument.
	\begin{lemma}
		Let $\C$ be a $(M, n, \lambda_1, \lambda_2)$ DI code for AWGN channel $\mathcal{G}$. Then the minimal Euclidean distance is lower bounded by $g(\lambda_1+\lambda_2)$, where $g\;:\;(0, 1)\to \mathbb{R}$ is some function, such that $g(x)\to\infty$ as $x\to 0$.
	\end{lemma}
	\begin{proof}[Proof of Lemma]
		
		Consider two codewords $u_1$ and $u_2$ such that the distance between them is minimal and equals $d$.
		Let $D_1$ be a decoding region for the first identity. Let $f_Z(x)$ be a density of $n$-dimensional normal distribution with zero mean and a covariance matrix $\sigma^2\cdot Id$.
		Then 
		$$
		\int\limits_{D_1}f_Z(\y-\u_1) d\y\geq 1-\lambda_1,
		$$
		and
		$$
		\int\limits_{D_1}f_Z(\y-\u_2) d\y \leq \lambda_2.
		$$
		
		Since the density depends only on the distance from the center, we can assume without loss of generality that $\mathbf{u}_1=(0, 0, \ldots, 0)$, $\mathbf{u}_2=(d, 0, 0, \ldots, 0)$, $d>0$.
		Consider difference
		\begin{equation}\label{eq::awgn. int is big}
			\int\limits_{D_1}(f_Z(\y-\u_1)-f_Z(\y-\u_2)) d\y\geq 1-(\lambda_1+\lambda_2). 
		\end{equation}
		
		We are going to prove some upper bound for this expression. We can estimate the expression by the integral over $\mathbb{R}^n$ in the following manner
		\begin{align}
			&\int\limits_{D_1}(f_Z(\y-\u_1)-f_Z(\y-\u_2)) d\y\\<
			&\int\limits_{\mathbb{R}^n}(f_Z(\y-\u_1)-f_Z(\y-\u_2))^+d\y,
		\end{align}
		where $(x)^+=\max(x, 0)$. Using the definition of the density $f_Z(\x)$ we rewrite the previous expression as follows
		\begin{align}
			&\int\limits_{\mathbb{R}^n}(f_Z(\y-\u_1)-f_Z(\y-\u_2))^+d\y\\
			&=\int\limits_{\mathbb{R}^n}
			\left(\frac{1}{2\pi\sigma^2}\right)^{\frac{n-1}{2}}
			\exp\left(-\frac{\sum\limits_{i=2}^ny_i^2}{2\sigma^2}\right)\\
			&\cdot
			\left(\frac{1}{\sqrt{2\pi\sigma^2}}
			\left(\exp\left(-\frac{y_1^2}{2\sigma^2}\right)-\exp\left(-\frac{(y_1-d)^2}{2\sigma^2}\right)\right)\right)^+d\y\\
			&=
			\int\limits_{\mathbb{R}^{n-1}}\left(\frac{1}{2\pi\sigma^2}\right)^{\frac{n-1}{2}}
			\cdot\exp\left(-\frac{\sum\limits_{i=2}^ny_i^2}{2\sigma^2}\right)dy_2dy_3\ldots dy_n\\
			&\cdot
			\int\limits_{-\infty}^{\infty}\left(\frac{1}{\sqrt{2\pi\sigma^2}}
			\left(\exp\left(-\frac{y_1^2}{2\sigma^2}\right)-\exp\left(-\frac{(y_1-d)^2}{2\sigma^2}\right)\right)\right)^+dy_1.
		\end{align}
		The first multiplier is just an integral of the density of $n-1$ dimensional normal random variable and thus equals 1. Notice that the function $\exp\left(-\frac{y_1^2}{2\sigma^2}\right)-\exp\left(-\frac{(y_1-d)^2}{2\sigma^2}\right)$ is positive if and only if $y_1<d/2$, so we can rewrite the second multiplier as follows.
		\begin{align}
			&\frac{1}{\sqrt{2\pi\sigma^2}}\int\limits_{-\infty}^{\infty}
			\left(\exp\left(-\frac{y_1^2}{2\sigma^2}\right)-\exp\left(-\frac{(y_1-d)^2}{2\sigma^2}\right)\right)^+dy_1\\
			&=
			\frac{1}{\sqrt{2\pi\sigma^2}}\int\limits_{-\infty}^{d/2}
			\left(\exp\left(-\frac{y_1^2}{2\sigma^2}\right)-\exp\left(-\frac{(y_1-d)^2}{2\sigma^2}\right)\right)^+dy_1\\
			&<\frac{1}{\sqrt{2\pi\sigma^2}}\int\limits_{-\infty}^{d/2}
			\exp\left(-\frac{y_1^2}{2\sigma^2}\right)dy_1=
			P(Z_1<d/2).
		\end{align}
		So, we proved that
		\begin{equation}
			\int\limits_{D_1}(f_Z(\y-\u_1)-f_Z(\y-\u_2)) dy< P(Z_1<d/2).
		\end{equation}
		From \eqref{eq::awgn. int is big} we obtain
		\begin{equation}
			P(Z_1<d/2)>1-(\lambda_1+\lambda_2).
		\end{equation}
		which gives a lower bound on the minimal distance $d$. Notice that when $\lambda_1+ \lambda_2\to 0$ the distance $d$ should tend to infinity.
	\end{proof}
	
	If the minimal Euclidean distance between any codewords of the code $\C$ is greater than $d$, then $M$ balls with radius $r=d/2$ and centers in the codewords are pairwise disjoint. Moreover, since $\norm{\u_i}_2^2\leq An$ for all codewords $\u_i$ we conclude that all these balls lie inside a bigger ball of radius $\sqrt{An}+r$. Then the number of codewords $M$ can be estimated as a ratio of volumes of $n$-dimensional balls with radiuses $\sqrt{An}+r$ and $r$.
	$$
	M\leq \left(\frac{\sqrt{An}+r}{r}\right)^n.
	$$
	From inequality
	$$
	2^{Rn\log n}\leq \left(\frac{\sqrt{An}+r}{r}\right)^n
	$$
	we obtain
	$$
	R\leq \frac{\log\frac{\sqrt{An}+r}{r}}{\log n}=\frac12+o(1).
	$$

	\section{Proofs of Theorems~\ref{th::slow fading} and~\ref{th::fast fading}}\label{sec::fading with csi}
	
	
	
	Throughout the proofs we will use the following statement.

\begin{proposition}\label{pr::packing implies code}
    Assume that for every $n>n_0$ there exists a code $\mathcal{U}(n)=\{\u_i\}_{i\in[M]}$, $\u_i\in \mathbb{R}^n$, and the distance between any two codewords is at least $n^{1/4+b}$ for some positive constant $b$. Then it is possible to define decoding regions $\mathcal{D}(n)$ in such a way that $(\mathcal{U}(n), \mathcal{D}(n))$ would be a DI identification code for an AWGN channel with $\lambda_1(n), \lambda_2(n)\to 0$ as $n\to\infty$.
\end{proposition}

A similar proposition was proved in the paper\cite{salariseddigh2021deterministic} but the authors additionally assumed that $\norm{\u_i}_2^2\leq An$. We will need the proposition without this assumption.
\begin{proof}[Proof of Proposition~\ref{pr::packing implies code}]
    We define a decoding region $D_i$ as a ball with the center $\mathbf{u}_i$ and radius $\sqrt{\sigma^2 n + \sqrt{n}\ln n}$. Then the error of the first type equals
    \begin{equation}        \mathbb{P}\left(\sum\limits_{i=1}^nz_i^2>\sigma^2 n + \sqrt{n}\ln n\right),
    \end{equation}
    which tends to $0$ due to Chebyshev's inequality.

    To estimate the error of the second type we prove the following proposition.
    \begin{proposition}\label{pr::est of 2nd type error}
        Let $\mathbf{v}\in\mathbb{R}^n$ be some point in $n$-dimensional space, $\norm{\mathbf{v}}_2^2\geq n^{1/2+a}$ for some constant $a>0$. Let $z_i$, $i=1,\ldots,n$, be a sequence of independent random variables with distribution $\mathcal{N}(0, \sigma^2)$. Then 
        \begin{equation}
            \mathbb{P}\left(\norm{\mathbf{v}+\mathbf{z}}_2^2\leq \sigma^2 n + \sqrt{n}\cdot f(n)\right)\to 0,
        \end{equation}
        where $f(n)$ is an arbitrary function such that $f(n)=o(n^a)$.
    \end{proposition}
    Application of this proposition gives us the fact that the error of the second type tens to $0$. 
    \begin{proof}[Proof of Proposition~\ref{pr::est of 2nd type error}]

Define a random variable $\xi$ and a sequence of random variables $\eta_i$ as follows
\begin{align}
\eta_i&=(z_i+v_i)^2,\\
\xi&=\sum\limits_{i=1}^n\eta_i.
\end{align}

We are going to compute the mathematical expectation and the variance of $\xi$.
The mathematical expectation of $\eta_{i}$ equals
\begin{align}
    \mathbb{E}\eta_{i}&=\sigma^2 +  v_i^2.
\end{align}
Then
\begin{align}
    \mathbb{E} \xi = n\sigma^2  +\norm{\v}_2^2. 
\end{align}

The variance can be computed as $\text{Var }\xi=\sum\limits_{i=1}^n\text{Var }\eta_{i}$, since all random variables $\eta_{i}$ are mutually independent.
\begin{align}
    \mathbb{E}\eta_{i}^2&=
    \mathbb{E}(z_i^4+4z_i^3v_i+6z_i^2v_i^2
    +4z_iv_i^3+v_i^4)\\
    &=3\sigma^4+6\sigma^2v_i^2
    +v_i^4.\\
   (\mathbb{E}\eta_{i})^2
    &=
    \sigma^4+2\sigma^2v_i^2+v_i^4.
\end{align}
The variance of $\eta_i$ equals
\begin{align}
    \text{Var }\eta_i&=2\sigma^4+4\sigma^2v_i^2.
\end{align}
Finally, compute $\text{Var }\xi$.
\begin{align}
    \text{Var }\xi &= \sum\limits_{i=1}^n\text{Var }\eta_i=2n\sigma^4+4\sigma^2\norm{\mathbf{v}}_2^2.
\end{align}

Now we estimate the probability that the random variable $\xi$ is smaller than or equal to $\sigma^2 n + \sqrt{n}\cdot f(n)$.

\begin{align}
&\mathbb{P}(\xi\leq \sigma^2 n + \sqrt{n}\cdot f(n))\\
&\leq \mathbb{P}(\mathbb{E}\xi - \xi\geq \mathbb{E}\xi-(\sigma^2 n + \sqrt{n}\cdot f(n)))\\   
&\leq \mathbb{P}(|\xi-\mathbb{E}\xi|\geq \mathbb{E}\xi-(\sigma^2 n + \sqrt{n}\cdot f(n)))\\
&\leq \mathbb{P}(|\xi-\mathbb{E}\xi|\geq \norm{\mathbf{v}}_2^2-\sqrt{n} f(n)))\\
&= \mathbb{P}(|\xi-\mathbb{E}\xi|\geq \norm{\mathbf{v}}_2^2(1+o(1)))\\
&\leq \frac{\text{Var }\xi}{\norm{\mathbf{v}}_2^4}(1+o(1))\to 0.
\end{align}
    \end{proof}
\end{proof}

	\subsection{Proof of Theorem~\ref{th::slow fading}}

	We start the proof from the case $\eta<\mathbb{P}(h=0)$.
	Since $\eta<\mathbb{P}(h=0)$ then any set $H_0$ such that $\mathbb{P}(h\in H_0)\geq 1-\eta$ should contain zero. Then
	
	\begin{align}
		P_{e, 1}(i, 0)\leq \sup\limits_{h\in H_0}P_{e, 1}(i, h)&\leq \lambda_1,\\
		P_{e, 2}(i, j, 0)\leq \sup\limits_{h\in H_0}P_{e, 2}(i, j, h)&\leq \lambda_2,\\
	\end{align}
	From definitions of the errors of the first and second types it follows that
	$$
	P_{e, 1}(1, 0)+P_{e, 2}(2, 1, 0)=1,
	$$
	which means that $\lambda_1+\lambda_2\geq 1$.
	It means, that not only the capacity equals $0$, but in such case the number of codewords can't be bigger than 1.
	
	Now we proceed to the case $ \mathbb{P}(h=0)<\eta<1$.
	The upper bound follows from Theorem~\ref{th::upper bound for gaussian channel} for Gaussian channel with AWGN. Indeed, consider any sequence of DI codes $\mathcal{U}(n)=\{\u_i\}_{i\in[M]}$, $\u_i\in \mathbb{R}^n$ for slow fading channel under input constraints $A$ with lengths tending to infinity and the errors of the first and the second types tending to zero. 
	Recall that the errors of the first and the second types are defined as supremums of errors over $h\in H_0$. It means that we can pick some realization $B$ of a random variable $h$, $B\in H_0$. Then for this $B$ the errors of the first and the second type would be at most  $\lambda_1$ and $\lambda_2$ correspondingly. Then the code $B\cdot \u_i$ can be used as a code for AWGN channel under constraints $\norm{\x}_2^2\leq AB^2n$. The upper bound on the capacity of this code implies the same upper bound on the capacity of the original code for the slow fading channel.
	
	Now we prove the lower bound. Fix arbitrary $\eta$, $ \mathbb{P}(h=0)<\eta<1$.
	Introduce the threshold $T_n=\frac{1}{\ln n}$. Denote the probability that the absolute value of the fading coefficient is less than $T_n$ as $p(n)$, i.e.
	\begin{equation}
		p(n)=\mathbb{P}(h=0)+\mathbb{P}(h\in (-T_n, 0)\cup (0, T_n)).
	\end{equation}
	Note that $\mathbb{P}(h\in (-T_n, 0)\cup (0, T_n))$ monotonically tends to $0$, since 
	$\bigcap\limits_{n=1}^{\infty}\{\omega :h\in (-T_n, 0)\cup (0, T_n)\}=\emptyset$. Then
	$p_n$ monotonically decreases to $\mathbb{P}(h=0)$ as $n\to\infty$. Choose $n_0$ such that $p(n_0)<\eta$. Define $H_0=\mathbb{R}\setminus (-T_{n_0}, T_{n_0})$. Then $\mathbb{P}(h\in H_0)\geq 1-\eta$. Fix positive constant $a>0$. From Theorem~\ref{th::construction} we know that for any $n>n_0(a)$ there exists a code with polynomial encoding and construction algorithms $\mathcal{U}_n=\{\u_i\}_{i\in[M]}$, $\u_i\in \mathbb{R}^n$ of size $M\geq n^{(1/4-2a)n}$, such that $\norm{\u_i}^2\leq An$, and the distance between any two codewords is at least $n^{1/4+a}$. We are going to use this code to transmit identities. For any $h\in H_0$ the absolute value of the fading coefficient is at least $T_n$. In this case the code $\mathcal{U}'(n)=\{\u_i\cdot T_n\}_{i\in [M]}$ has the distance at least $n^{1/4+b}$ for $b=a/2$ for big enough $n$. Then we can use Proposition~\ref{pr::packing implies code} and conclude that the decoding regions can be defined in such a way that the code $(\mathcal{U}'(n), \mathcal{D}_h)$ would be a DI code for AWGN channel with $\lambda_1(n), \lambda_2(n)\to 0$ as $n\to\infty$. If we use the same decoding regions for our initial fading channel then the error probabilities can be upper bounded as $\lambda_1(n)$ and $\lambda_2(n)$ for all $h\in H_0$, hence, the probabilities of errors tend to zero.

	\subsection{Proof of Theorem~\ref{th::fast fading}}
	
	\begin{proof}[Proof of claim 1]
		The upper bound again follows from Theorem~\ref{th::upper bound for gaussian channel} for Gaussian channel with AWGN. Indeed, consider any sequence of DI codes $\mathcal{U}(n)=\{\u_i\}_{i\in[M]}$, $\u_i\in \mathbb{R}^n$ for fast fading channel under input constraints $A$ with lengths tending to infinity and the errors of the first and the second types tending to zero. Define a code $U'=\{\u_i'=\text{diag}(\mathbf{h})\u_i\}$. The mathematical expectation of $\norm{\u_i'}_2^2$ is equal to $\norm{\u_i}_2^2\mathbb{E}h^2$.
		From Markov's inequality, we conclude that 
		\begin{equation}\label{eq::norm of new codeword}
			\mathbb{P}(\norm{\u_i'}_2^2>Bn)<1/2
		\end{equation}
		for $B = 2A\mathbb{E} h^2$. Introduce a random variable $\zeta$ which is equal to the number of codewords $\u'$, which satisfy power constraints $\norm{\u'}_2^2\leq Bn$. 
		From~\eqref{eq::norm of new codeword} it follows that $\mathbb{E}\zeta\geq M / 2$. Again, from Markov's inequality, 
		\begin{equation}
			\mathbb{P}(\zeta>M/4)=1-\mathbb{P}(M-\zeta\geq 3M/4)\geq 1- \frac{\mathbb{E}(M-\zeta)}{3M/4}\geq \frac{1}{3}.
		\end{equation}
		Then for big enough $n$ there exists a vector $\mathbf{h}$,  such that $\zeta>M/4$ and the errors of the first and second type for such specific $\mathbf{h}$ are smaller than $3\lambda_1$ and $3\lambda_2$. Then we use the subcode of $U'$ which consists of $\u_i'$ that satisfies the power constraint $\norm{\u'}^2\leq Bn$. This code has size $\geq M/4$ and can be used as DI code for AWGN channel. The upper bound on the identification capacity of AWGN channel implies the same upper bound on the identification capacity for fast fading channel.
	\end{proof}
	\begin{proof}[Proof of claim 2]
		Now we prove the lower bound for the case $\mathbb{P}(h=0) = 0$. It is possible to choose a positive number $T_n$ such that 
		\begin{equation}
			\mathbb{P}(|h|< T_n) < \varepsilon_n=\frac{1}{\ln n}.
		\end{equation}
		Call a fading coefficient bad if its absolute value is smaller than $T_n$. Call a vector $\mathbf{h}$ of $n$ fading coefficients bad, if at least $2n\varepsilon_n$ of the coefficients are bad. Note that the probability that the vector $\mathbf{h}$ is bad tends to $0$. Indeed the number of bad coefficients $\xi$ is a Binomial random variable with parameters $n$ and $\varepsilon_n$. The probability that the vector $\mathbf{h}$ is bad can be estimated with the help of Hoeffding's inequality~\cite{hoeffding1994probability}
		\begin{equation}
			p(n)=\mathbb{P}(\xi-\mathbb{E} \xi > \mathbb{E} \xi)\leq e^{-2(\mathbb{E} \xi)^2/n}=e^{-\frac{2n}{(\ln n)^2}}\to 0.
		\end{equation}

		\begin{proposition}\label{pr::packing with distance between projections}
			For any positive constant $a<1/8$ and any $n>n_0(a)$ there exists a code $\mathcal{U}(n)=\{\u_i\}_{i\in[M]}$, $u_i\in \mathbb{R}^n$ with $M\geq 2^{\left(\frac{1}{4}-2a\right)n\log_2n}$ such that $\norm{\u_i}_2^2\leq An$, and the distance between the projection of any two codewords into any set of at least $\mu_n\cdot n$ coordinates is at least $n^{\alpha}$, $\alpha=1/2-1/(4\mu_n)+a/\mu_n$ for any sequence $\mu_n\geq \mu_{min}>0$.
		\end{proposition}
		
		\begin{proof}
			Take a random code of size $2M$ where each coordinate is chosen independently according to a normal distribution $N(0, A')$, $A'<A$. 
			Call a codeword $\u$ bad if $\norm{\u}_2^2> An$. Note that the probability that the codeword is bad tends to $0$. Indeed, let $\xi=\norm{\u}^2$. Then $\mathbb{E}\xi = A'n$, $\text{Var }\xi = 2A'^2n$. From Chebyshev's inequality it follows that
			\begin{equation}
				\mathbb{P}(\xi > An)\leq \mathbb{P}(|\xi-\mathbb{E}\xi|>(A-A')n)\leq \frac{\text{Var }\xi}{(A-A')^2n^2}=o(1).
			\end{equation}
			Let $\eta$ be a squared distance between the projection of two fixed codewords $\u_1$ and $\u_2$ onto a fixed set of $\mu_n \cdot n$ coordinates $C$. Then $\eta$ can be represented as a sum
			\begin{equation}
			    \eta = \sum\limits_{i=1}^{\mu_n n}\eta_i^2,
			\end{equation}
			where $\eta_i$ are i.i.d. random variables with normal distribution $N(0, 2A')$, i.e. $\frac{\eta}{2A'}$ is a chi-squared random variable.
			
			Denote the probability $\mathbb{P}(\eta\leq n^{2\alpha})$ as $p_1$. We upper bound the probability $p_1$ by using Chernoff's inequality\cite{chernoff1952measure}
			\begin{equation}
				p_1\leq \inf_{t<0}\mathbb{E}e^{t\eta}e^{-tn^{2\alpha}}=
				\inf_{t<0}\mathbb{E}e^{2A't\frac{\eta}{2A'}}e^{-tn^{2\alpha}}.
			\end{equation}
			
			Moment generating function $\mathbb{E}e^{t'\frac{\eta}{2A'}}$ for chi-squared random variable $\chi^2(k)$ is known to be equal to $(1-2t')^{-k/2}$ for $t'<1/2$, so
			
			\begin{equation}
				p_1\leq 
				\inf_{t<0}(1-4A't)^{-\mu_nn/2}e^{-tn^{2\alpha}}.
			\end{equation}
			
			Infinum is attained at $t=\frac{1}{4A'}-\frac{\mu_nn}{2n^{2\alpha}}$ and gives the bound
			
			\begin{equation}
				p_1\leq 
				\left(\frac{2\mu_nnA'}{n^{2\alpha}}\right)^{-\mu_nn/2}e^{-\frac{n^{2\alpha}}{4A'}+\frac{\mu_nn}{2}}.
			\end{equation}
			Taking into account that $\alpha<1/2$ we get
			\begin{equation}
				p_1\leq n^{n(-(1-2\alpha)\mu_n/2+o(1))}=n^{-n(1/4-a+o(1))}.
			\end{equation}

			The probability that for some fixed codeword $\u_i$ there exists another codeword $\u_j$ and the set of $\mu_n \cdot n$ coordinates $C$ such that the distance between the projections of these codewords onto the set $c$ is smaller than $n^{\alpha}$ is upper bounded by
			\begin{equation}
				2^n\cdot2Mp_1=
				n^{n(-(1/4-a)(1-\alpha)+(1/4-2a)+o(1))}.
			\end{equation}
			
			For any $a>0$ we can choose $\alpha>0$ to be small enough such that $-(1/4-a)(1-\alpha)+(1/4-2a)<0$. In that case, the probability that our fixed codeword has a small distance from some other codeword for some projection tends to $0$. So, we can apply the expurgation method and obtain the code of size $M$.
			
		\end{proof}
		Fix a positive $a<1/8$ and consider a code from Proposition~\ref{pr::packing with distance between projections} for $\mu_n=1-2\varepsilon_n$. This code $\mathcal{U}$ will be used as an ID code for the fast fading channel.  Consider a code $\mathcal{U}'(n)=\{\text{diag}(\mathbf{h})\u_i\}$. If the vector $\mathbf{h}$ is not bad then the distance of this code is at least $n^{1/2-1/(4\mu)+a/\mu_n}$, which is greater than $n^{1/4+a/2}$ for $n$ big enough. So we can apply Proposition~\ref{pr::packing implies code} to construct decoding regions with $\lambda_1(n), \lambda_2(n)\to 0$ as $n\to\infty$. If we use the same decoding regions for our initial fading channel then the error probabilities can be upper bounded as $\lambda_1(n)+p(n)$ and $\lambda_2(n)+p(n)$, hence, the probabilities of errors tend to zero.
	\end{proof}
	\begin{proof}[Proof of claim 3]
		Define a random variable $\xi$ to be equal to $\norm{\y - \text{diag}(\mathbf{h})\u_i}_2^2$.
		
		Define decoding regions as balls with a center 
		$\text{diag}(\mathbf{h})\u_i$ and radius $\sqrt{\mathbb{E}\xi+\sqrt{\text{Var }\xi\ln n}}$. Then Chebyshev's inequality guarantees us that the error of the first type tends to zero.
		
		Compute $\mathbb{E}\xi$ and $\text{Var }\xi$.
		
		\begin{equation}
			\mathbb{E}\xi = n\sigma^2.
		\end{equation}
		\begin{equation}
			\text{Var }\xi = 2n\sigma^4.
		\end{equation}
		
		Now we define a random variable $\xi'$ to be equal to the squared distance from the output codeword for identity $j'$ to the center of decoding ball $D_j$.

		\begin{align}
			\xi'&=\sum\limits_{i=1}^n\eta_i'=\sum\limits_{i=1}^n(h_ix_i'+z_i-h_ix_i)^2\\
			&=\sum\limits_{i=1}^n(z_i+h_i(x_i'-x_i))^2,
		\end{align}
		i.e.
		\begin{align}
			\eta_i'=(z_i+h_i(x_i'-x_i))^2.
		\end{align}

		We are going to compute the mathematical expectation and variance of $\xi'$.
		The mathematical expectation of $\eta_{i}'$ equals
		\begin{align}
			\mathbb{E}\eta_{i}'&=\sigma^2 +  \mathbb{E}h^2(x_i'-x_i)^2.
		\end{align}
		Then
		\begin{align}
			\mathbb{E} \xi' = n\sigma^2  +\mathbb{E}h^2\norm{\x-\x'}_2^2. 
		\end{align}
		
		The variance can be computed as $\text{Var }\xi'=\sum\limits_{i=1}^n\text{Var }\eta_{i}'$, since all random variables $\eta_{i}'$ are mutually independent.
		\begin{align}
			\mathbb{E}\eta_{i}'^2&=
			\mathbb{E}(z_i^4+4z_i^3h_{i}(x_i'-x_i)+6z_i^2h_{i}^2(x_i'-x_i)^2\\&
			+4z_ih_{i}^3(x_i'-x_i)^3+h_{i}^4(x_i'-x_i)^4)\\
			&=
			\mathbb{E}(z_i^4 + 6z_i^2h_{i}^2(x_i'-x_i)^2
			+h_{i}^4(x_i'-x_i)^4)
			\\&
			=3\sigma^4+6\sigma^2\mathbb{E}h^2(x_i'-x_i)^2
			\\&+
			\mathbb{E}h^4(x_i'-x_i))^4.
		\end{align}
		Now we compute $(\mathbb{E}\eta_{i}')^2$.
		\begin{align}
			(\mathbb{E}\eta_{i}')^2
			&=
			\sigma^4+(\mathbb{E}h^2)^2(x_i'-x_i)^4+2\sigma^2\mathbb{E}h^2(x_i'-x_i)^2.
		\end{align}
		The variance of $\eta_i'$ equals
		\begin{align}
			\text{Var }\eta_i'&=2\sigma^4+4\sigma^2\mathbb{E}h^2(x_i'-x_i)^2+\text{Var }h^2(x_i'-x_i)^4.
		\end{align}
		Finally, compute $\text{Var }\xi'$.
		\begin{align}
			\text{Var }\xi' &= \sum\limits_{i=1}^n\text{Var }\eta_i'\\    &=2n\sigma^4+4\sigma^2\mathbb{E}h^2\norm{\x'-\x}_2^2+\text{Var }h^2\norm{\x'-\x}_4^4.
		\end{align}

		\begin{restatable}{proposition}{packingPropositionNew}\label{pr::packing exists2}
			For any positive constants $a$, $A$, $B$, $a<1/8$, and any $n>n_0(a)$ there exists a code $\mathcal{U}(n)=\{\u_i\}_{i\in[M]}$, $\u_i\in \mathbb{R}^n$ of size $M\geq 2^{\left(\frac{1}{4}-2a\right)n\log_2n}$
			with the following properties.
			\begin{enumerate}
				\item  $\norm{\u_i}_2^2\leq An$ for all $i$.
    \vspace{1mm}
				\item  $\norm{\u_i}_4^4\leq Bn$ for all $i$.
    \vspace{1mm}
				\item $\norm{\u_i-\u_j}_2\geq n^{1/4+a}$ for all $i\neq j$.
			\end{enumerate}
            Moreover, code $\mathcal{U}(n)$ can be efficiently constructed and encoded, i.e., the time complexity of construction and encoding are polynomial.
		\end{restatable}
		\begin{proof}
			The existence of such code follows from Theorem~\ref{th::construction}.
		\end{proof}
		If we use such code then the variance of $\xi'$ is linear.
		
		The error of the second type can be upped bounded as follows
		\begin{align}
			&\mathbb{P}(\xi' < \mathbb{E}\xi+\sqrt{\text{Var }\xi\ln n})\\
			&=\mathbb{P}(\mathbb{E}\xi' - \xi' >\mathbb{E}\xi'- \mathbb{E}\xi-\sqrt{\text{Var }\xi\ln n})
			\\&\leq
			\mathbb{P}(|\xi' - \mathbb{E}\xi'| >\mathbb{E}h^2\norm{\x-\x'}_2^2-\sqrt{\text{Var }\xi\ln n})    \\
			&=
			\mathbb{P}(|\xi' - \mathbb{E}\xi'| >\mathbb{E}h^2n^{1/2+2a}(1+o(1)))\\&
			\leq \frac{\text{Var }\xi'}{(\mathbb{E}h^2)^2n^{1+4a}(1+o(1))}=o(1).
		\end{align}
		
	\end{proof}

	\section{Proof of Theorem~\ref{th::fast fading without CSI}}\label{sec::fast fading without csi}
	
	To prove this theorem we need the following proposition.
	
	\begin{proposition}\label{pr::packing exists3}
		For any positive constants $a, A', A$, $a<1/8$, $A'<A$, and any $n>n_0(a)$ there exists a code $\mathcal{U}(n)=\{\u_i\}_{i\in[M]}$, $\u_i\in \mathbb{R}^n$ of size $M\geq 2^{\left(\frac{1}{4}-2a\right)n\log_2n}$
		with the following properties.
		\begin{enumerate}
			\item  $\norm{\u_i}_2^2\leq An$ for all $i$.
   \vspace{1mm}
			\item  $\norm{\u_i}_4^4\leq 3A^2n$ for all $i$.
   \vspace{1mm}
			\item  $\norm{|\u_i}_2^2 - A'n|\leq \sqrt{n}\ln n$ for all $i$.
   \vspace{1mm}
			\item $\norm{\u_i-\u_j}_2\geq n^{1/4+a}$ for all $i\neq j$.
		\end{enumerate}
	\end{proposition}
	\begin{proof}[Proof of Proposition~\ref{pr::packing exists3}]
		The proof is a straightforward application of the probabilistic method with expurgation, so we provide only a sketch. Take a random code of size $2M$, where all coordinates of all codewords are independent random variables with normal distribution $N(0, A')$.
		Call a codeword $\u_i$ bad, if one of 3 first properties is violated for $\u_i$, or if there exists another codeword $\u_j$, such that the fourth property is violated for $\u_i$ and $\u_j$. For the first 3 properties estimations are trivial, for the fourth it is done in the same way as in Proposition~\ref{pr::packing with distance between projections}. Then prove that the probability that a codeword is bad is less than $1/2$. Delete all bad codewords from the code. The mathematical expectation of the final code is at least $M$, so the code with desired properties exists.
		
	\end{proof}

	We are going to use the code from Proposition~\ref{pr::packing exists3}.
	Say that we send codeword $\x$ through the fading channel
	\begin{equation}
		y_t=h_tx_t+z_t.
	\end{equation}
	Then the mathematical expectation of the result is
	\begin{equation}
		\mathbb{E}y_t=c\cdot x_t,
	\end{equation}
	where $c=\mathbb{E}h_i$. Define a random variable $\xi=\norm{\y-\mathbb{E}\y}_2^2$.
	\begin{align}
		\xi&=\sum\limits_{i=1}^n\eta_i=\sum\limits_{i=1}^n(h_ix_i+z_i-cx_i)^2\\
		&=\sum\limits_{i=1}^n(z_i^2+x_i^2(h_i-c)^2+2 x_iz_i(h_i-c)).   
	\end{align}
	The mathematical expectation of $\eta_i$ equals
	\begin{align}
		\mathbb{E}\eta_i=\sigma^2 + x_i^2\text{Var }h,
	\end{align}
	so
	\begin{align}
		\mathbb{E} \xi = n\sigma^2 + \text{Var }h\cdot \norm{\x}_2^2 = \Theta(n). 
	\end{align}
	
	The variance can be computed as $\text{Var }\xi=\sum\limits_{i=1}^n\text{Var }\eta_i$, since all random variables $\eta_i$ are mutually independent.
	
	\begin{align}
		\mathbb{E}\eta_i^2&=
		\mathbb{E}z_i^4
		+\mathbb{E}x_i^4(h_i-c)^4
		+4\mathbb{E}x_i^2(h_i-c)^2z_i^2
		\\&
		+2\mathbb{E}z_i^2x_i^2(h_i-c)^2
		\\&=
		3\sigma^4+x_i^4\mathbb{E}(h-c)^4
		+6x_i^2\sigma^2\text{Var } h.
	\end{align}
	
	Now we compute $(\mathbb{E}\eta_i)^2$.
	\begin{align}
		(\mathbb{E}\eta_i)^2
		&=
		\sigma^4+2x_i^2\sigma^2 \text{Var }h+(\text{Var }h)^2x_i^4.
	\end{align}
	The variance of $\eta_i$ equals
	\begin{align}
		\text{Var }\eta_i=2\sigma^4+x_i^4\text{Var } (h-c)^2+4x_i^2\sigma^2 \text{Var }h.
	\end{align}
	Finally, compute $\text{Var }\xi$.
	\begin{align}
		\text{Var }\xi &= \sum\limits_{i=1}^n\text{Var }\eta_i\\
		&=2n\sigma^4 + 4\sigma^2\text{Var }h\norm{\x}_2^2+\norm{\x}_4^4\text{Var }(h - c)^2.
	\end{align}
	We note that the variance is a linear function of $n$.
	
	We can define a decoding region for the identity with codeword $\x$ as a ball with the center $c\x$ and radius $\sqrt{\mathbb{E}\xi + \sqrt{\text{Var }\xi \ln n}}$, which will guarantee us that the error of the first type tends to zero due to  Chebyshev's inequality.
	
	Now we need to estimate the error of the second type.
	To do that we define a new random variable $\xi'=\norm{\y'-\mathbb{E}\y}^2$, which is equal to the square of the distance from the received codeword $\y'$ to the center of the decoding ball for another identity.
	\begin{align}
		\xi'&=\sum\limits_{i=1}^n\eta_i'=\sum\limits_{i=1}^n(h_ix_i'+z_i-cx_i)^2\\
		&=\sum\limits_{i=1}^n(z_i+x_i(h_i-c)+h_i(x_i'-x_i))^2,
	\end{align}
	i.e.
	\begin{align}
		\eta_i'=(z_i+x_i(h_i-c)+h_i(x_i'-x_i))^2.
	\end{align}
	The mathematical expectation of $\eta_i'$ equals
	\begin{align}
		\mathbb{E}\eta_i'&=\sigma^2 + x_i^2\text{Var }h + \mathbb{E}h^2(x_i'-x_i)^2
		\\&+2x_i(x_i'-x_i)\text{Var }h\\
		&=
		\sigma^2 + c^2(x_i'-x_i)^2\\
		&+\text{Var }h(x_i^2+(x_i'-x_i)^2+2x_i(x_i'-x_i))\\
		&=\sigma^2 + c^2(x_i'-x_i)^2+\text{Var }hx_i'^2.
	\end{align}
	So,
	\begin{align}
		\mathbb{E} \xi' = n\sigma^2 + \text{Var }h\cdot \norm{\x'}_2^2 + c^2\norm{\x-\x'}_2^2. 
	\end{align}
	The variance can be computed as $\text{Var }\xi'=\sum\limits_{i=1}^n\text{Var }\eta_i'$, since all random variables $\eta_i'$ are mutually independent.
	\begin{align}
		\mathbb{E}\eta_i'^2&=
		\mathbb{E}(z_i^4+4z_i^3(h_ix_i'-cx_i)+6z_i^2(h_ix_i'-cx_i)^2\\&
		+4z_i(h_ix_i'-cx_i)^3+(h_ix_i'-cx_i)^4)\\
		&=
		\mathbb{E}(z_i^4 + 6z_i^2(h_ix_i'-cx_i)^2
		+(h_ix_i'-cx_i)^4)
		\\&
		=3\sigma^4+6\sigma^2(\mathbb{E}h^2x_i'^2-2c^2x_i'x_i+c^2x_i^2)
		\\&+
		\mathbb{E}((h_i-c)x_i'+c(x_i'-x_i))^4
		\\&=
		3\sigma^4+6\sigma^2\text{Var }hx_i'^2+6\sigma^2c^2(x_i-x_i')^2\\
		&+\mathbb{E}(h-c)^4x_i'^4+4\mathbb{E}(h-c)^3cx_i'^3(x_i'-x_i)
		\\&+
		6\text{Var }hx_i'^2c^2(x_i'-x_i)^2+c^4(x_i'-x_i)^4.
	\end{align}
	Now we compute $(\mathbb{E}\eta_i')^2$.
	\begin{align}
		(\mathbb{E}\eta_i')^2
		&=
		\sigma^4+c^4(x_i'-x_i)^4+(\text{Var }h)^2x_i'^4\\
		&+2\sigma^2c^2(x_i'-x_i)^2+2\sigma^2\text{Var }hx_i'^2\\
		&+2c^2\text{Var }hx_i'^2(x_i'-x_i)^2.
	\end{align}
	The variance of $\eta_i'$ equals
	\begin{align}
		\text{Var }\eta_i'&=2\sigma^4+4\sigma^2\text{Var } hx_i'^2+4\sigma^2c^2(x_i-x_i')^2
		\\&+
		x_i'^4\text{Var }(h-c)^2+4\mathbb{E}(h-c)^3cx_i'^3(x_i'-x_i) 
		\\&+
		4c^2\text{Var } hx_i'^2(x_i'-x_i)^2.
	\end{align}
	Finally, compute $\text{Var }\xi'$.
	\begin{align}
		\text{Var }\xi' &= \sum\limits_{i=1}^n\text{Var }\eta_i'\\
		&=2n\sigma^4+4\sigma^2\text{Var } h\norm{\x'}_2^2+4\sigma^2c^2\norm{\x-\x'}_2^2
		\\&+
		\norm{\x'}_4^4\text{Var }(h-c)^2+4\mathbb{E}(h-c)^3c\sum\limits_{i=1}^n x_i'^3(x_i'-x_i) 
		\\&+
		4c^2\text{Var } h\sum\limits_{i=1}^n x_i'^2(x_i'-x_i)^2.
	\end{align}
	Note that the variance is linear since any monomial $x_i^kx_i'^{4-k}$ for $k\in [0, 4]$ can be upper bounded by $x_i^4+x_i'^4$.
	
	The error of the second type can be upped bounded as follows
	\begin{align}
		&\mathbb{P}(\xi' < \mathbb{E}\xi+\sqrt{\text{Var }\xi\ln n})\\
		&=\mathbb{P}(\mathbb{E}\xi' - \xi' >\mathbb{E}\xi'- \mathbb{E}\xi-\sqrt{\text{Var }\xi\ln n})
		\\&\leq
		\mathbb{P}(|\xi' - \mathbb{E}\xi'| >\text{Var h}(\norm{\x'}_2^2 - \norm{\x}_2^2)\\
		&+c^2\norm{\x-\x'}_2^2-\sqrt{\text{Var }\xi\ln n})  \\
		&\leq
		\mathbb{P}(|\xi' - \mathbb{E}\xi'| >c^2n^{1/2+2a}(1+o(1)))\\&
		\leq \frac{\text{Var }\xi'}{c^4n^{1+4a}(1+o(1))}=o(1).
	\end{align}

	\section{Conclusion}\label{sec::conclusion}
	In this paper, we consider deterministic identification codes for the slow and fast fading channels. For the slow and fast fading channels with side information, we proved lower bounds for a wide class of probability distribution of fading coefficients. Unlike previously known results our theorems give lower bounds when $0$ belongs to the closure of the set of values of fading coefficients, which is often the case for practical distributions. Moreover, the lower bound on the capacity of the fast fading channel holds even when the probability $\mathbb{P}(h=0)$ is positive. We also presented a lower bound for the fast fading channel without side information. 
	
	Additionally, we provide a construction of the code, which achieves the best-known lower bounds on the capacities of the slow and fast fading channel with side information. This code has efficient construction and encoding algorithms and can be implemented in practice.
	
	In future work, it would be interesting to consider a more practical model with a complex-valued code, fading coefficients, and noise. We expect that the results can be adapted without significant problems. Another interesting task is to get rid of some conditions in our theorems. For example, condition $\mathbb{E} h\ne 0$ in Theorem~\ref{th::fast fading without CSI} forbids us to apply the theorem for Rayleigh and Nakagami distribution, even though it is still applicable for Rician. Proving the analog of Theorem~\ref{th::fast fading without CSI} without the condition $\mathbb{E} h\ne 0$ would be an interesting and practically important result.
	
	\bibliographystyle{IEEEtran}
	\bibliography{mybib}
	
\end{document}